\pdfoutput=1
\newif\ifFull
\Fullfalse
\ifFull
\documentclass[11pt]{article}
\usepackage[margin=1in]{geometry}
\else
\documentclass{llncs}
\fi
\usepackage{graphicx}
\usepackage{enumerate}
\usepackage[nocompress]{cite}
\usepackage{color}
\usepackage{subfig}
\usepackage[noend]{algorithmic}
\usepackage{url}

\usepackage{verbatim}

\ifFull\else
\renewcommand{\subsection}[1]{\paragraph{\textbf{#1}.}}
\fi

\title{Capturing Lombardi Flow in\\
       Orthogonal Drawings by Minimizing \\
       the Number of Segments}

\author{
Md.~Jawaherul~Alam
\and
Michael~Dillencourt
\and
Michael~T.~Goodrich
}

\institute{
Department of Computer Science, University of California,
Irvine, CA, USA
}

\newcommand{\NP}{NP}
\newcommand{\Deg}{{\rm deg}}
\newcommand{\emb}{\mathcal{E}}
\newcommand{\closure}{closure}

\begin{document}

\maketitle

\begin{abstract}
Inspired by the artwork of Mark Lombardi, we study the problem of
constructing orthogonal drawings where a small number of horizontal
and vertical line segments covers all vertices. 
\ifFull
This gives the impression
of a visual flow across the drawing for a viewer and the aesthetic
quality of such visual flow is evidenced by recent user studies on
orthogonal drawings. 
\fi
We study two problems on
orthogonal drawings of planar graphs, one that minimizes the total number 
of line segments and another that minimizes the number of line
segments that cover all the vertices.
We show that the first problem can be solved by a non-trivial
modification of the
flow-network orthogonal bend-minimization algorithm of Tamassia,
resulting in a polynomial-time algorithm.
We show that the second problem is 
\NP-hard even for planar graphs with maximum degree 3.
Given this result,
we then address this second optimization
problem for trees and series-parallel graphs with maximum
degree 3. For both graph
classes, we give polynomial-time 
algorithms for upward orthogonal drawings with the
minimum number of segments covering the vertices.
\end{abstract}

\section{Introduction}

The American artist Mark Lombardi~\cite{lombardi} 
produced many hand-drawn visualizations 
of social networks to represent conspiracy theories, 
with the image shown in Fig.~\ref{fig:lombardi} being one of his most
famous works. 
Many people have observed that his drawings have
strong aesthetic qualities and high readability.
As a result,
there has been considerable work in the graph drawing literature 
inspired by his work
(e.g., see~\cite{DEGKN12,ABKKKW14,subcubic,DEGKL11,PHNK12,LN12,CCGKT11}).
For example, Duncan \textit{et al.}~\cite{DEGKN12} 
observed that in the art work by Lombardi
edges tend to be nearly evenly spaced around each vertex, and they
introduced the notion of a \emph{Lombardi drawing}, where the edges
incident on each vertex are drawn to have perfect angular resolution. 
Our approach in this paper, however, is to focus instead on a different
aesthetic quality evident in some of the works of Lombardi.

\begin{figure}[hbt!]
\centering
\includegraphics[width=\textwidth]{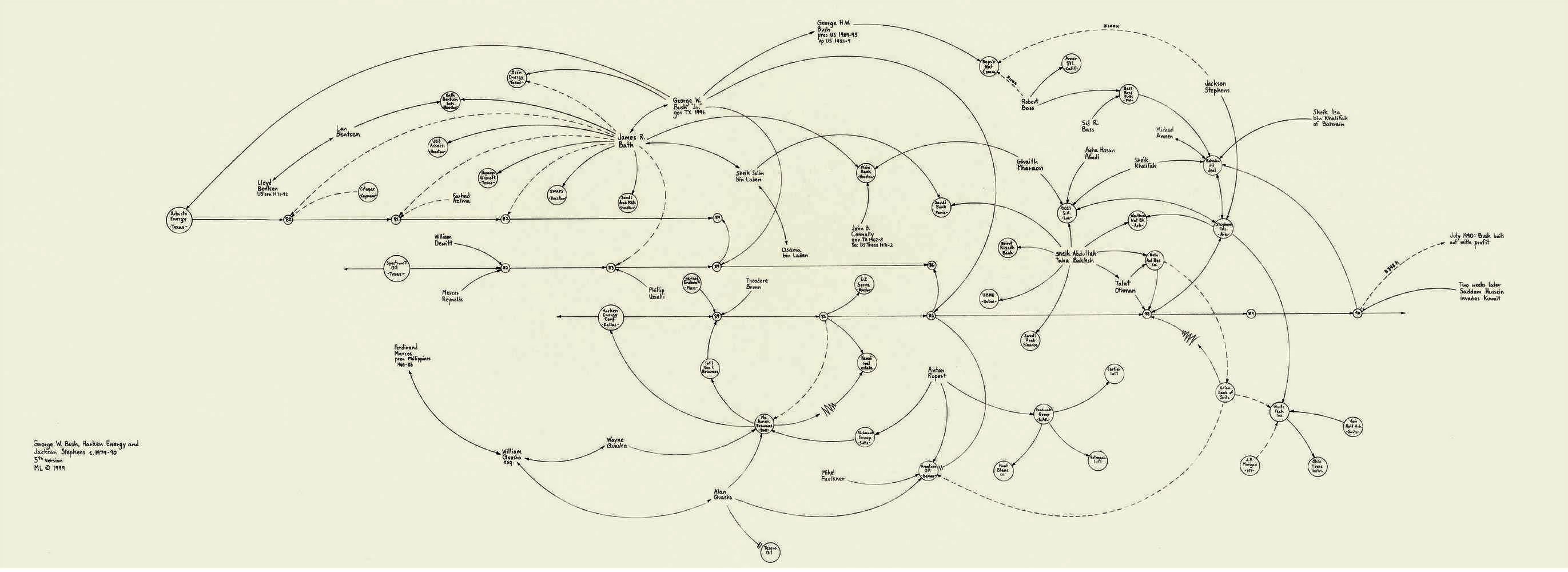}
\caption{Mark Lombardi: George W. Bush, Harken Energy, and Jackson 
Stevens c.~1979--90, 1999. Graphite on paper, 
$20 \times 44$ inches~\cite{lombardi}.}
\label{fig:lombardi}
\end{figure}

Another aesthetic property
prevalent in Lombardi's art work is that he 
tends to place many vertices on consecutive
stretches of linear or circular segments that go across the whole
drawings; see Fig.~\ref{fig:lombardi}. This creates a metaphor of
a ``visual flow'' across a drawing. 
Inspired by this property, we
are interested in this paper in
a new aesthetic criteria for orthogonal drawings of graphs.
Namely,
we would like to place the vertices of a graph on a small number
of contiguous line segments in the drawing.


As is well-known in the graph drawing literature,
\textit{orthogonal drawing} is a common drawing paradigm
for visualizing planar graphs in numerous applications, ranging from
VLSI circuit design~\cite{Va81,Le80}, architectural
floor-planning~\cite{LM81}, UML diagram design~\cite{See97}, and a
popular visualization technique used in network layout systems,
e.g. see~\cite{yed,graphviz,OGDF}. 
\ifFull
In such a drawing of a planar graph,
each vertex is drawn as a point and each edge is drawn (without
crossing) as an orthogonal (axis-aligned) polyline segment between
the corresponding points. 
\fi
For readability and aesthetics of such a
drawing, traditional orthogonal layout algorithms~\cite{BETT09,KW01,NS}
try to minimize the number of bends on the edges in the drawing.
However, a recent empirical study by Kieffer {\it et al.}~\cite{hola} 
concluded that orthogonal
layouts generated by traditional automatic algorithms that focus
primarily on bend and area minimization lack
some aesthetic qualities compared to manual drawings produced by humans. 
Their study instead suggests that, like in several of the works
by Lombardi, humans prefer drawings
with linear ``flow'' that connect chains of adjacent vertices.
Our specific interest in this paper is to study such 
``Lombardi flow'' for orthogonal graph drawings.

\subsection{Related Work}

In a seminal work, Tamassia~\cite{Tam87} showed that for 
a \textit{plane graph} 
(i.e., a planar graph with a fixed planar embedding) with
the maximum degree $4$, an orthogonal drawing with the minimal
number of edge bends can be computed in $O(n^2 \log n)$ time, by
solving a minimum-cost maximum-flow problem. This time complexity
was later reduced to $O(n^{7/4}\sqrt{\log n})$ by Garg and
Tamassia~\cite{GT96}, to $O(n^{1.5})$ by Cornelsen and
Karrenbauer~\cite{cornel}, and to linear time for plane graphs
with maximum degree 3 by Rahman and Nishizeki~\cite{RN02}. In
contrast, if the embedding is not given, the problem becomes
NP-hard~\cite{GT01}. Bertolazzi \textit{et al.}~\cite{BBD00} proposed
a branch-and-bound approach for minimizing the number of bends over
all possible embeddings. There have also been some work on restricting
the number of bends locally, e.g., on minimizing the maximum number of
bends per edge~\cite{BK98, BKRW14} in an orthogonal drawing instead
of the total number of bends.
Biedl and Kant~\cite{BK98} gave a linear-time
algorithm for constructing an orthogonal drawing with at most two
bends per edge for any graph with degree at most four except for
the octahedron. 
Orthogonal drawings of graphs with vertices of degree higher
than four have also been studied~\cite{BBD00,foessmei,TBB88}.
Bl{\"a}sius \textit{et al.}~\cite{BKRW14} 
give a polynomial-time algorithm to decide whether
a plane graph has an orthogonal drawing with at most one bend per
edge. 
There has also been some research on 
\textit{rectangular drawings}, where each face is 
rectangle~\cite{AKM-ortho,RNG04}. 
This is not the same as a minimum-segment drawing, 
though, since
not every planar graph with maximum degree 4 admits a rectangular drawing,
and even when such a drawing exists,
it does not necessarily give a drawing
with the minimum number of segments containing 
all vertices, e.g., see~\ref{fig:example01} in the appendix.
In contrast, minimizing the number of non-orthogonal 
segments in a drawing of a graph has been
previously addressed, 
\ifFull
in the context of straight-line drawings, for outer-planar
graphs, 2-trees, planar 3-trees, and 3-connected plane
graphs~\cite{DESW07}, and for series-parallel graphs with the maximum
degree 3~\cite{SAAR08}. 
\else
in the context of straight-line drawings~\cite{DESW07,SAAR08,DMNW13,MNBR13,DM14,IMS15}. Minimizing the number of circular arcs to draw planar graphs have also been studied~\cite{BR15,Schulz15,HKMS16}.
\fi
However to the best of our knowledge, this
problem has never been studied before in the context of orthogonal
drawings.

\subsection{Our Results}
We study two problems for orthogonal drawings of planar graphs. 
These problems
attempt to capture and formalize the concept of ``Lombardi flow:''
\begin{enumerate}
\item
A \textit{minimum segment orthogonal drawing}, or \textit{MSO-drawing},
of a planar graph $G$ is an orthogonal drawing of $G$ with the
minimum number of segments. 
Given a combinatorial embedding of a degree-4
planar graph, $G$, the MSO-drawing problem is to produce such a drawing of $G$.
\item
A \textit{minimum segment cover orthogonal drawing}, or 
\textit{MSCO-drawing}, of $G$
is one with the smallest set of segments covering all vertices of
$G$. 
Given a combinatorial embedding of a degree-4
planar graph, $G$, the MSCO-drawing problem is to produce such a drawing of $G$.
\end{enumerate}

In this paper, we present the following results.
\begin{itemize}
\item 
We provide an interesting modification to the network-flow algorithm
by Tamassia~\cite{Tam87}
to give a polynomial-time algorithm 
for the MSO-drawing problem.
\item 
We show that the MSCO-drawing problem is \NP-hard 
even for plane graphs with maximum degree 3.
\item 
For trees and series-parallel graphs with maximum degree 3,
 we give polynomial-time algorithms for upward 
orthogonal drawings with the minimum number of segments covering the
vertices. This solves the upward versions of the MSCO-drawing problem
in polynomial time for these types of graphs.
\end{itemize}

\section{Preliminaries}

An \textit{orthogonal drawing} of a planar graph, $G$, is a planar
drawing of $G$ such that each vertex is drawn as a point and each
edge is drawn as a orthogonal (axis-aligned) path between the
corresponding points. For an orthogonal drawing, it is generally
assumed that the input graph is planar with maximum degree 4, and
that for each vertex, at most one edge is incident in its four
orthogonal \textit{ports}: left, right, top, bottom. 
Although models incorporating vertices with degree more than 4 were
introduced later (e.g., by F{\"o}{\ss}meier and Kaufmann~\cite{foessmei}),
in this article we assume that the input graph is always a planar
graph with maximum degree 4.
A \emph{maximal line-segment} of an orthogonal drawing, $\Gamma$, is a 
set of line-segments (all horizontal or all vertical) $l_1, l_2,
\ldots, l_k$, such that $k$ is as large as possible 
so that, for $i\in 1, 2, \ldots, k-1$, $l_i$ and
$l_{i+1}$ have a common end point. In this paper we simply use the
term \textit{segment} to refer to a maximal line-segment.

Tamassia~\cite{Tam87} showed that an orthogonal drawing, $\Gamma$,
of a biconnected plane graph\ifFull\footnote{A planar graph with a fixed
     combinatorial planar embedding is known as a \emph{plane} graph.}
     \fi, $G$, can be described by augmenting the embedding
of $G$ with the angles at the bends (\emph{bend angles}) and the
angles between pairs of consecutive edges around the vertices of
$G$ (\emph{vertex angles}). Specifically, an \emph{angle assignment}
is a mapping from the set $\{\pi/2,\pi,3\pi/2\}$ to the angles of
$G$, where each angle is assigned exactly one value. An angle
assignment of $G$ can precisely describe the shape of $\Gamma$,
although it does not specify edge lengths. Given an angle assignment
$\Phi$, one can test if $\Phi$ corresponds to an orthogonal drawing
of $G$ by Lemma~\ref{lem:tamassia}, which is implied from~\cite{Tam87};
also see~\cite{AKM-ortho}.

\begin{lemma}[Tamassia~\cite{Tam87}]
\label{lem:tamassia}
An angle assignment $\Phi$ for a plane graph $G$ 
 corresponds to an orthogonal drawing of $G$ if and only if $\Phi$ satisfies the following
 conditions $(P_1$--$P_2)$.
\begin{enumerate}[$($P$_1)$]
	\item The sum of the assigned angles around each vertex $v$ in $G$ is $2\pi$.
	\item The number of $\pi/2$ angle around each face $f$ is exactly $4$ more (resp. $4$ less) than
		the number of $3\pi/2$ angle around $f$, if $f$ is an inner face (resp. the outer face).
\end{enumerate}
Given an angle assignment $\Phi$ satisfying $(P_1$--$P_2)$, one can obtain an orthogonal drawing of $G$ (i.e., the exact coordinates for the vertices) in linear time.
\end{lemma}

\section{Minimum Segment Orthogonal Drawing}

In this section, we address the problem of finding an orthogonal
drawing of an embedded planar graph, $G$, with maximum degree 4 using
a minimum number of
segments. 
We call such a drawing a \textit{minimum-segment orthogonal drawing}, 
or an \textit{MSO}-drawing. We show that by a modification
and generalization of a the network-flow algorithm by Tamassia~\cite{Tam87},
we can give a polynomial-time algorithm to find an MSO-drawing of
$G$.
In addition, this generalization also enables an interactive
trade-off adjustment in the orthogonal drawing between bends at a degree-2
vertex and edge bends, if that is desired.

We first briefly review the minimum-cost flow-network formulation
of Tamassia~\cite{Tam87} for computing minimum-bend orthogonal
drawings of plane graphs. We then describe how this network can be
modified and generalized to compute drawings with the minimum number
of segments. We first assume that the input graph is connected and
the degree of each vertex in the input graph is at least~$2$. In
Section~\ref{sec:general} we consider a general planar graph with
maximum degree~4.

\subsection{Network Flow-Formulation by Tamassia}

Given a plane graph $G$ with maximum degree 4, the algorithm
by Tamassia~\cite{Tam87}
constructs a flow network $H$ in order to find an orthogonal layout
with the minimum total number of bends on the edges.  In Tamassia's
network, $H$, there are \textit{boundary-vertices}, representing the
vertices of $G$, and \textit{face-vertices}, representing the faces
of $G$; see Fig.~\ref{fig:tamassia}(a). The edges of $H$ are the
bidirectional edges of the dual graph of $G$ (dashed grey edges in
Fig.~\ref{fig:tamassia}(a), called \textit{dual edges}) and the
edges from each boundary-vertex to its incident face-vertices (solid
grey edges, called \textit{angle edges}). Each boundary-vertex, $v$,
is a source that produces $4$ units. Each face-vertex
for a face $f$ is a sink that consumes either 
$2\, \Deg(f) - 4$ units (for inner faces) 
or $2\, \Deg(f) + 4$ units (for the outer
face), where $\Deg(f)$ is the number of vertices on the face $f$. The cost
of an edge is $1$ unit if it connects two face-vertices (i.e., a
dual edge), and $0$ otherwise. A min-cost max-flow in this network
corresponds to an orthogonal drawing of $G$, as follows. A flow of
$t\in\{1, 2, 3\}$ units from a boundary-vertex to a face-vertex
determines a $t\pi/2$ assignment to the corresponding angle in $G$.
A flow of $t$ units through some dual edge corresponds to $t$ bends
in the corresponding edge of $G$; see Fig.~\ref{fig:tamassia}(b).
Using this minimum-cost flow-network, Tamassia~\cite{Tam87} gave
an $O(n^2 \log n)$-time algorithm for orthogonal drawing with minimum
number of bends. Cornelsen and Karrenbauer~\cite{cornel} used the
same network but improved the running time to $O(n^{1.5})$ with a
faster min-cost max-flow algorithm for this planar network.

\subsection{Flow Formulation for Minimum-Segment Orthogonal Drawings}
\label{sec:modification}

We modify and generalize the flow network by Tamassia to
obtain an orthogonal drawing with a minimum number of segments.
In our modification, we change the consumption and production of
the flows in the boundary and face vertices, modify some edge
directions and costs, and interpret the flows on the edges slightly
differently. 

In our modified network, we again have the
\textit{boundary-vertices}, \textit{face-vertices} and the
bidirectional \textit{dual edges}. However instead of the
unidirectional angle edges from each boundary-vertices to its
incident face-vertices, we have bidirectional \textit{angle edges}
between each boundary-vertices and its incident face-vertices; see
Fig.~\ref{fig:modification}(a). Each face-vertex for an inner face
$f$ is a source with a production of $4$ units, while the face-vertex
for the outer face is a sink with a consumption of $4$ units. Each
boundary-vertex for a vertex $v$ with degree 3 or 4 in $G$ is a sink
with a consumption of $2$, or $4$ units, respectively, while
boundary-vertices corresponding to vertices of degree 2 is neither a
source nor a sink. 
Thus the consumption of a boundary-vertex for a vertex $v$ 
is $2\, \Deg(v) - 4$, where $\Deg(v)$ denotes the degree of $v$ in $G$. 
We set the capacity of flow on each direction of an angle edge to be
$1$ unit, while the dual edges are uncapacitated in both directions.
Finally, the cost of an edge is $c_1$ if it is a dual edge, and $c_2$
if it is an angle edge, for some constant $c_1$, $c_2$. We 
discuss below how to make suitable assignments to $c_1$ and $c_2$ to obtain a
minimum-segment orthogonal drawing.

We interpret the flow on different edges of $H$ to obtain an angle
assignment for $G$ as follows. For each (bidirectional) angle edge
$e$, a flow of $0$ unit gives a straight ($\pi$) angle assignment
to the corresponding angle, while a flow of $1$ unit from the
face-vertex to the boundary vertex (resp. from the boundary vertex
to the face-vertex) in $e$ gives a $\pi/2$ (resp. $3\pi/2$) assignment
to the angle. A flow of $t$ units through each dual edge again
corresponds to $t$ bends in the corresponding edge of $G$. We first
show that a valid flow in the modified flow network $H$ corresponds
to a valid angle assignment for an orthogonal layout of $G$.

\begin{lemma}
\label{lem:flow} 
Each valid assignment of flows on the edges of $H$ corresponds to an angle assignment for $G$, satisfying Conditions ($P_1$--$P_2$) of Lemma~\ref{lem:tamassia}, and vice versa.
\end{lemma}
\begin{proof} 
Consider first a valid flow $\mathcal{F}$ in $H$. We show that this
gives a desired angle assignment to $G$. For each vertex $v$ with
degree $4$ in $G$, its corresponding boundary-vertex has four
incident angle edges and a consumption of $4$ unit. Thus the
boundary-vertex has to receive $1$ unit of flow in each angle,
giving a $\pi/2$ assignment to each angle around $v$. Similarly, a
boundary-vertex corresponding to a vertex of degree $3$ has
to receive $1$ unit of flow from exactly two of its three incident
angle edges. Thus, each degree-$3$ vertex has two $\pi/2$ and one
$\pi$ assignment around itself. Finally consider the boundary-vertex
$v$ for each degree-$2$ vertex. Either there is no flow in either of the
angle edges incident to $v$ (which gives two $\pi$ angles around
the degree-$2$ vertex), or $v$ receives $1$ unit of flow in one
incident edge and gives back the flow in the other (which gives a
$\pi/2$ and a $3\pi/2$ assignments to the two angles). Thus in all
cases, for each vertex of $G$, the summation of the assigned angles
around it is $2\pi$, satisfying Conditions $P_1$.

For each face-vertex, each outgoing flow (of $1$
unit) corresponds to a $\pi/2$ angle (on a vertex or an edge bend),
and each incoming flow (of $1$ unit) corresponds to a $3\pi/2$
angle. Since each face-vertex corresponding to an inner face is a
source with a production of $4$ units, and the face-vertex for the
outer face is a sink with a consumption of $4$ units, Condition $P_2$
follows.

Conversely for each valid angle assignment $\Phi$ on $G$, we can
construct a valid flow assignment on the edges of $H$ as follows.
By Condition $P_1$, for each vertex of degree $4$, all four of its
incident angles are $\pi/2$, thus the corresponding boundary-vertex
receives a total of $4$ units of flow as necessary. Similarly, for
each degree-$3$ vertex, exactly two of its incident angles is
$\pi/2$, providing the sources of the necessary $2$ units of incoming
flow for the boundary-vertex. A degree-$2$ vertex can either have
two incident $\pi$ angles, or a $\pi/2$ and a $3/\pi/2$ angles. In
the first case, the boundary-vertex has no incoming or outgoing
flow, while in the second case, it has equal ($1$ unit) incoming
and outgoing flow. Again By Condition $P_2$, the face-vertex for
each inner face sends $4$ more unit of flow than what it receives,
while the face-vertex for the outer face receives $4$ more unit of
flows than what it sends.
\qed
\end{proof}

We now discuss suitable values of the edge costs $c_1$, $c_2$. Let
$\mathcal{F}$ be a valid flow in $H$, and let $\Gamma$ be the
corresponding orthogonal layout of $G$. Note that each edge bend
incurs a cost of $c_1$, while each $\pi/2$ and $3\pi/2$ angle
assignment around a vertex in $\Gamma$ incurs a cost of $c_2$ in
$\mathcal{F}$. In any angle assignment for $G$, there are exactly
four $\pi/2$ angles around each degree-$4$ vertex, and there are
exactly two $\pi/2$ angles around each degree-$3$ vertex in $G$.
On the other hand, for each degree-$2$ vertex, there are either two
$\pi$ angles around it, or there is a $\pi/2$ and a $3\pi/2$ angle
around it. Call a degree-$2$ vertex a \textit{vertex-bend} in $\Phi$
if it has a $\pi/2$ and a $3\pi/2$ angle around it. Thus the total
cost in $\mathcal{F}$ is $c_1(4\#(V_4)+3\#(V_3)+2\#(B_v))+c_2\#(B_e)$,
where $\#(V_4)$, $\#(V_3)$, $\#(B_v)$, $\#(B_e)$ denotes the number
of degree-$4$ vertices in $G$, the number of degree-$3$ vertices,
the number of vertex bends in $\Gamma$, and the number of edge-bends
in $\Gamma$, respectively.  Minimizing the cost in this network for
different values of $c_1$ and $c_2$ optimizes different properties
of the resulting orthogonal layout. (One might verify that $c_1=1$,
$c_2=0$ gives an orthogonal layout with minimum number of total
edge bends, similar to Tamassia~\cite{Tam87}). Lemma~\ref{lem:MSO}
gives a suitable value for $c_1$, $c_2$ to minimize the number of
segments in the corresponding orthogonal layout.

\begin{lemma}
\label{lem:MSO} A minimum-cost flow in $H$ for the values of $c_1=1/2$, $c_2=1$ gives an angle assignment of $G$ corresponding to a minimum segment drawing of $G$.
\end{lemma}
\begin{proof}  
Let $\mathcal{F}$ be a valid flow in $H$ and let $\Gamma$ be the
corresponding orthogonal layout of $G$. Let $\#(V_4)$, $\#(V_3)$,
$\#(B_v)$, $\#(B_e)$ denote the number of degree-$4$ vertices in
$G$, the number of degree-$3$ vertices, the number of vertex bends
in $\Gamma$, and the number of edge-bends in $\Gamma$, respectively.
We count the number of segments in $\Gamma$. Each segment in $\Gamma$
has two end points, and each end-point of a segment in $\Gamma$ is
a degree-$3$ vertex or a vertex-bend or an edge-bend. 
Conversely, each vertex or edge bend corresponds to the end point of
exactly two segments, while each degree-$3$ vertex corresponds to
the end point of exactly one segment. Thus the number of segments
$\#(s)$ in $\Gamma$ is $(\#(V_3)+2\#(B_v)+2\#(B_e))/2$. Now with
$c_1=1/2$, $c_2=1$, the cost function of $H$ is (from the previous
paragraph):
$(4\#(V_4)+3\#(V_3)+2\#(B_v))/2+\#(B_e)=(2\#(V_4)+\#(V_3))+(\#(V_3)+2\#(B_v)+2\#(B_e))/2=k+\#(s)$,
where $(2\#(V_4)+\#(V_3))=k$ is a constant for any orthogonal drawing
$\Gamma$ of $G$. Hence for $c_1=1/2$, $c_2=1$, minimizing the
cost function also minimizes the number of segments in $\Gamma$.
\qed
\end{proof}

Note that $H$ is a planar network with size linear to the number of vertices in~$G$. Cornelsen and Karrenbauer~\cite{cornel} showed that if the total cost is linear, then a minimum-cost flow on a planar network can be computed in $O(n^{1.5}\log n)$ time. Thus summarizing the results in this section, we have the following theorem.

\begin{theorem}
\label{th:MSO} For an embedded connected $n$-vertex planar graph,
$G$, with maximum degree 4 and no vertex of degree $1$,
a minimum segment orthogonal drawing of $G$ can be computed 
in $O(n^{1.5}\log n)$ time.
\end{theorem}

\subsection{Computing Exact Drawing: Aligning Multiple Segments Together}

In Section~\ref{sec:modification}, we show how to compute an orthogonal
representation for a planar graph by computing the angle assignment
for the angles around each vertex and by specifying the edge bends.
This is a topological description of a topologically equivalent
class of orthogonal drawings, but such an orthogonal drawing is
dimensionless. In order to create an exact orthogonal drawing from
such a  representation, Tamassia~\cite{Tam87} first augments the
layout into a rectangular layout by decomposing each non-rectangular
face into rectangles with the introduction of dummy edges and
vertices. Then integer lengths are assigned to the sides of each
rectangular face to obtain the exact coordinate. 

As an additional extension to Tamassia's algorithm, 
we give here an algorithm to augment an orthogonal representation into
a rectangular layout such that the maximum number of pairs of
segments align with each other, thereby extending the ``visual
flow'' across faces in a drawing. Let $\Gamma$ be an orthogonal
layout of a plane graph $G$ with the maximum degree 4, and let
$\Gamma^*$ be rectangular layout containing $\Gamma$.  
A pair of vertices $v_1$, $v_2$ are said
to be \textit{aligned together} if they belong to the same maximal
line-segment in $\Gamma^*$, but do not belong to the same maximal
line-segment in $\Gamma$.  
and both $l_1$, $l_2$ belong to the same maximal line-segment in
$\Gamma^*$.  Formally we address the following problem:

\begin{itemize} 
\item \textbf{(Augmentation):} Given an orthogonal layout $\Gamma$ of a plane graph $G$ with the maximum degree 4, construct a rectangular layout $\Gamma^*$ such that $\Gamma$ is contained in $\Gamma^*$ and the maximum number of pairs of vertices from $G$ are aligned together in $\Gamma^*$.
\end{itemize}

We solve the Augmentation problem
 by computing a maximum matching for a bipartite graph constructed for each non-rectangular face of $\Gamma$ independently. For each face $f$ we construct a bipartite graph $B(f)$ whose edges represent potential pairs of vertices which can be aligned together inside $f$. We need to select a maximum matching on $B(f)$ in such a way that all the pairs of vertices corresponding to these matching edges can be aligned together simultaneously. In Appendix~\ref{sec:app-proofs} we show that such a matching can be computed in polynomial time, using the Hopcroft-Karp algorithm~\cite{HK73} for example. Doing this for each non-rectangular faces in $\Gamma$ gives an augmentation $\Gamma^*$ of $\Gamma$ into a rectangular layout with the maximum pairs of vertices aligned together; see Appendix~\ref{sec:app-proofs}.

\begin{theorem}
\label{th:augmentation} The Augmentation problem 
can be solved in polynomial time.
\end{theorem}

\subsection{General Planar Graphs with Maximum Degree~4}
\label{sec:general}

In Section~\ref{sec:modification} we considered planar graphs with maximum degree~4 and with no degree-1 vertex. Here we consider all planar graphs with the maximum degree~4. Let $G$ be planar graph with the maximum degree~4. By repetitively deleting vertices of degree 1 from $G$, we get a subgraph $G'$ of $G$ with no vertex of degree 1 (possibly with a single vertex, in which case $G$ is a tree). Call $G'$ the \textit{\closure} of $G$. Call the vertices of $G'$ with at least one adjacent vertex in $G$, not belonging to $G'$ the \textit{connection vertices} of $G$. From $G$, delete all vertices of $G'$.
Each connected component of the resulting graph is then a tree. Given a general planar graph $G$ with the maximum degree~4, we therefore draw the {\closure} of $G$ using the algorithm in Section~\ref{sec:modification}, and attach the remaining tree parts to the connection vertices. For drawing a rooted tree $T$ (rooted at a connection vertex) with the maximum degree 4, we first draw an arbitrary root-to-leaf path in a single segment. We then repetitively take a shortest path from a leaf to vertex already drawn, and draw the path in a segment to complete a minimum segment drawing of $T$. Finally we attach the drawing of each tree part to the drawing of the {\closure} of $G$ at the corresponding connection vertex to get a minimum segment drawing for the embedding given by the drawing; see Fig.~\ref{fig:example03}(b). The details are in Appendix~\ref{sec:app-proofs}.

\begin{theorem}
\label{th:general} For an $n$-vertex planar graph, $G$, with maximum degree 4, an orthogonal drawing $\Gamma$ of $G$ can be obtained in polynomial time such that $\Gamma$ is a minimum segment orthogonal drawing for the embedding of $G$ in $\Gamma$.
\end{theorem}

\section{Minimum Segment Cover Orthogonal Drawing}

Let $\Gamma$ be an orthogonal drawing for a planar graph, $G$, with
maximum degree 4. A
set of segments $S$ in $\Gamma$ is said to \textit{cover} $G$, if
each vertex in $G$ is on at least one of the segments in $S$. The
\textit{segment-cover number} of $\Gamma$ is the minimum cardinality
of a set of segments covering $G$. Given a planar graph $G$, 
with maximum degree 4, a
\textit{minimum segment cover orthogonal drawing} or \textit{MSCO}-drawing
of $G$ is an orthogonal drawing that has the minimum cardinality
of a set of segments covering $G$. The number of segments in such
an optimal set is also called the \textit{segment-cover number} of
$G$.

\subsection{\NP-Hardness}
We begin our discussion in this section by showing
that finding a minimum segment cover orthogonal drawing
is \NP-hard, even for planar graphs with maximum degree 3. 
\ifFull
We then
address the problem for trees and for a subclass of such graphs:
series-parallel graphs with maximum degree 3. For both graph
classes, we give algorithms for an orthogonal drawing of any
series-parallel planar graph with maximum degree 3 
that achieves the optimum segment-cover
number with the restriction that the drawing is upward.
\fi

\begin{theorem}
\label{thm:np}
Finding a minimum segment cover orthogonal drawing for a 
planar graph $G$ is \NP-hard, even if $G$ is a planar graph with
maximum degree 3.
\end{theorem}
\begin{proof} 
\ifFull
We prove the \NP-hardness of the problem by a reduction from the
\textit{Hamiltonian path problem}~\cite{Garey:1990:CIG:574848}, 
which asks whether a given graph
$G$ has a path that contains each vertex of $G$ exactly once (i.e.,
a \textit{Hamiltonian path}). 
The Hamiltonian path problem remains
\NP-complete for planar graphs with maximum degree 3~\cite{doi:10.1137/0205049}.
Let $G$ be such a planar graph. We show that $G$ has a Hamiltonian path
if and only if its segment-cover number is 1.

Clearly, if $G$ has an orthogonal drawing that contains a segment
covering all the vertices, then the path of $G$ corresponding to
this segment is a Hamiltonian path. Conversely if $G$ has a Hamiltonian
path $P$, then we can find an orthogonal drawing of $G$ with a
single segment covering all the vertices as follows. Take a fixed
planar embedding $\emb(G)$ of $G$, draw $P$ as a single segment
(say a horizontal segment) and route the remaining edges as orthogonal
polylines with necessary bends keeping the embedding $\emb(G)$
fixed. This gives a valid orthogonal drawing since each vertex has
a maximum degree 3, and hence from each vertex $v$, its incident
edges on $P$ are drawn to its left and right, while the only other
edge (if any) is drawn to its top or bottom without any overlap.
\else
We prove the \NP-hardness of the problem by a reduction from the
\textit{Hamiltonian path problem}~\cite{Garey:1990:CIG:574848};
please see the appendix for details.
\fi
\qed
\end{proof}

\subsection{Algorithm for Series-Parallel Graphs with Maximum Degree 3}
\label{sec:SP}

A {\em series-parallel} graph is a graph $G$ that has two terminals $s$ and
$t$, and either $G$ is an edge $(s,t)$, or it has been constructed
in one of the following ways:
(1) (Parallel combination)
$G$ consists of two or more 
series-parallel graphs that all have the terminals $s$ and $t$.  (2)
(Combination in series)
$G$ consists of two or more series-parallel graphs $G_1, G_2, \ldots, G_f$ such that each $G_i$, $i\in\{1, 2, \ldots, f\}$ has terminals $\{s_i,t_i\}$ with terminals $t_i=s_{i+1}$ for $i\in\{1, 2, \ldots, f\}$ and $s_1=s$, $t_f=t$.

The recursive construction of a series-parallel graph is represented by an \textit{SPQ-tree}, which is a rooted tree defined as follows. An SPQ-tree $\mathcal{T}$ for a series-parallel graph $G$ contains at most three types of nodes: $S$-, $P$-, and $Q$-nodes. Informally each edge of $G$ corresponds to a $Q$-node, while an $S$-node and a $P$-node denote a series and a parallel combination, respectively. If $G$ contains a single edge, then $\mathcal{T}$ consists of a single $Q$-node. Otherwise if $G$ is obtained by a parallel (resp. series) combination of two or more series-parallel graphs $G_1, G_2, \ldots, G_f$, then $\mathcal{T}$ is rooted at a $P$-node (resp. $S$-node) with the SPQ-trees for $G_1, G_2, \ldots, G_f$ as children subtrees of the root.

In this section we consider series-parallel graphs with maximum degree $3$. For each such graph, Lemma~\ref{lem:SP3} gives some properties of its SPQ-tree, which is implied from~\cite{SAAR08}.

\begin{lemma}~\cite{SAAR08}
\label{lem:SP3} Let $G$ be a series-parallel graph with the maximum degree 3 and let $\mathcal{T}$ be its SPQ-tree. Then the following two conditions hold:
\begin{itemize}
\item For each $S$-node $x$ of $\mathcal{T}$ with children nodes $x_1, x_2, \ldots, x_f$, (i) each $x_i$, $i\in\{1,2,\ldots, f\}$ is either a $P$-node or a $Q$-node, (ii) both $x_1$ and $x_f$ are $Q$-nodes, and (iii) if $x_i$ is a $P$-node for some $i\in\{2,\ldots, f-1\}$, then both $x_{i-1}$ and $x_{i+1}$ are $Q$-nodes.

\item Each non-root $P$-node has exactly two children and a root $P$-node (if any) has at most three children. At most one of the children of a $P$-node is a $Q$-node, and the remaining children are $S$-nodes.
\end{itemize}
\end{lemma}

A polyline $l$ is \textit{$y$-monotone} if for any horizontal line $h$, the intersection $l\cap h$ is a connected line-segment (possibly empty or a single point). Let $\mathcal{T}$ be the SPQ-tree of a series-parallel $G$ such that the terminals for the root node of $\mathcal{T}$ are $s$ and $t$. Then an orthogonal drawing $\Gamma$ of $G$ is \textit{upward} if any path from $s$ to $t$ in $G$ is drawn as a $y$-monotone polyline in $\Gamma$.

\begin{theorem}
\label{th:SP} Let $G$ be a series-parallel graph with the maximum degree 3 with the SPQ-tree $\mathcal{T}$ and let $\#(P^*)$ be the number of $P$-nodes in $\mathcal{T}$ with at least two $S$-nodes as children. If the root of $\mathcal{T}$ is a $P$-node with three $S$-nodes as children, then $\#(P^*)+2$ segments are necessary and sufficient to cover all vertices of $G$ in any upward orthogonal drawing of $G$; otherwise $\#(P^*)+1$ segments are necessary and sufficient.
\end{theorem}

We prove Theorem~\ref{th:SP} by constructing an algorithm that obtains an upward orthogonal drawing of $G$ with the desired segment cover number, and then proving that this bound is also necessary. Assume first that each $P$-node in the SPQ-tree $\mathcal{T}$ for $G$ (including the root, if it is a $P$-node) has exactly 2 children.

\begin{lemma}
\label{lem:SP-suff} Let $G$ be a series-parallel graph with the maximum degree 3 with the SPQ-tree $\mathcal{T}$ where each $P$-node has exactly 2 children. Let $\#(P^*)$ be the number of $P$-nodes in $\mathcal{T}$ with two $S$-nodes as children. Then there is an upward orthogonal drawing of $G$ with $\#(P^*)+1$ vertical segments covering all the vertices. 
\end{lemma}
\begin{proof}
We give an inductive construction of an upward orthogonal drawing of $G$. If $G$ contains only a single edge (i.e., the root $r$ of $\mathcal{T}$ is a $Q$-node), then we draw $G$ with a single vertical line segment. Otherwise, if $r$ is either an $S$-node or a $P$-node. If $r$ is an $S$-nodes with children $x_1, \ldots, x_f$, then let $G_1, \ldots, G_f$ be the series-parallel graphs corresponding to the subtrees of $\mathcal{T}$ rooted at $x_1, \ldots, x_f$, respectively. For $i\in\{1,2, \ldots, f\}$, let $s_i, t_i$ be the two terminals of $G_i$ and let $\#(P^*_i)$ denote the number of $P$-nodes with two children $S$-nodes in the subtree rooted at $x_i$. We then recursively draw upward orthogonal drawing $\Gamma(x_i)$ for $G_i$ with $\#(P^*_i)+1$ vertical segments covering all the vertices of $G_i$, for $i\in\{1,\ldots, f\}$. We then obtain an upward orthogonal drawing of $G$ by placing the drawings $\Gamma(x_1)$, $\Gamma(x_2)$, $\ldots$, $\Gamma(x_f)$ one after another from top to bottom and by identifying $t_i$ with $s_{i+1}$ for $i\in\{1,2,\ldots, f-1\}$; see Fig.~\ref{fig:SP}(a). This identification of vertices are possible since the drawings $\Gamma(x_i)$'s are upward (hence $s_i$ and $t_i$ are at the topmost and bottommost $y$-coordinates) for each $i\in\{1,\ldots, f\}$. Furthermore, by Lemma~\ref{lem:SP3}, for each $P$-node child $x_i$, both $x_{i-1}$ and $x_{i+1}$ are $Q$-nodes (and hence are drawn as a vertical line segment). Thus there is no edge overlap in the resulting drawing. Since each drawing $\Gamma(x_i)$ has $\#(P^*_i)+1$ vertical segments covering all the vertices of $G_i$, and for $i\in\{2,\ldots, f\}$, $\Gamma(x_i)$ shares exactly one of these vertical segments with $\Gamma(x_{i-1})$, the total vertical segments covering all the vertices of $G$ is $(\#(P^*_1)+1)+(\#(P^*_2))+\ldots+(\#(P^*_f)) = \#(P^*)+1$.

Finally if the root $r$ of $\mathcal{T}$ is a $P$-node with children $x_1$ and $x_2$, then let $G_1$, $G_2$ be the two series-parallel graphs corresponding to the two subtrees of $\mathcal{T}$ rooted at $x_1$ and $x_2$, respectively. For $i\in\{1,2\}$, let $s_i, t_i$ be the two terminals of $G_i$ and let $\#(P^*_i)$ denote the number of $P$-nodes with two children $S$-nodes in the subtree rooted at $x_i$. We recursively obtain and upward orthogonal drawings $\Gamma(x_1)$ and $\Gamma(x_2)$ of $G_1$, $G_2$, respectively. Let $x_1$ is an $S$ node. Then we place the two drawings $\Gamma(x_1)$, $\Gamma(x_2)$ side-by-side and identify $s_1$, $s_2$ and $t_1$, $t_2$ by making two bends at the very top edge and the very bottom edge in $\Gamma(x_2)$; see Fig.~\ref{fig:SP}(b). There is a only a single edge at the top and the bottom since $x_2$ is either a $Q$-node or an $S$-node, and if $x_2$ is an $S$-node, then by Lemma~\ref{lem:SP3}, its topmost and bottommost part contain a single vertical segment for its children $Q$-nodes. If $x_2$ is an $S$-node, then $\#(P^*) = \#(P^*_1) + \#(P^*_2) +1)$, and the total number of vertical segments covering all vertices of $G$ is $(\#(P^*_1)+1) + (\#(P^*_2)+1) = \#(P^*) +1$. Otherwise, The segment between the two bend points in $\Gamma(x_2)$ does not contain any vertex, and the total number of vertical segments covering all vertices of $G$ is $(\#(P^*_1)+1) = \#(P^*) +1$. \qed
\end{proof}

Fig.~\ref{fig:SP-illus} illustrates an example of a drawing obtained from the algorithm in the proof of Lemma~\ref{lem:SP-suff}. We are now ready to prove Theorem~\ref{th:SP}.

\noindent
\textit{Proof of Theorem~\ref{th:SP}.} If the root $r$ of $\mathcal{T}$ is not a $P$-node, or if it has only two children, then we find an upward orthogonal drawing of $G$ by the algorithm described in the proof of Lemma~\ref{lem:SP-suff}. Otherwise $r$ is a $P$-node with its three children $x_1$, $x_2$, $x_3$ $\mathcal{T}_1$ is the SPQ-tree obtained after deleting the subtree rooted at $x_3$ from $\mathcal{T}$. Also $\mathcal{T}_2$ is the SPQ-tree obtained after deleting the subtrees rooted at $x_1$, $x_2$ from $\mathcal{T}$. Let $G_1$, $G_2$ be the series-parallel graphs corresponding to these two SPQ-trees, respectively. We find an upward orthogonal drawings $\Gamma_1$, $\Gamma_2$ for $G_1$, $G_2$, place them side-by-side and identify their two pairs of terminals by making bends at the very top and the very bottom segments in $\Gamma_2$. This gives an upward orthogonal drawing of $G$ with the desired segment-cover number.

To see that the drawing achieves the optimal segment-cover number, observe the following. Let $x_p$ be an arbitrary non-root $P$-node in $\mathcal{T}$ with two $S$-node children $x_1$, $x_2$. Let $s$ and $t$ be the two terminals corresponding to $x_p$. Then deleting $\{s,t\}$ splits $G$ into exactly three components, two containing vertices of the two series parallel graphs $G_1$, $G_2$ corresponding to the subtrees rooted at $x_1$ and $x_2$, respectively, and the third containing vertices from neither $G_1$ nor $G_2$. In any upward orthogonal drawing of $G$, both $s$ and $t$ (or their incident edges) will be the end point of at least one segments that contains vertices of $G$. By Lemma~\ref{lem:SP3}, no two $P$-nodes share any terminals. Thus if $\#(P_n^*)$ be the number of non-root $P$-nodes, then there are at least $2\#(P_n^*)$ end-points of segments containing vertices, and thus at least $\#(P_n^*)$ number of segments containing vertices in any upward orthogonal drawing. Finally If the the root node $r$ is a $P$ node, then since the drawing is upward, none of the graph corresponding to its each child $S$-node and its descendants can share a segment covering vertices, hence the result follows. \qed

An upward drawing of a series-parallel graph with the maximum degree 3 is not always an MSCO-drawing; see Fig.~\ref{fig:SP-illus}(d). However for any outer-planar graph (which is a series-parallel graph), each $P$-node in the SPQ-tree has exactly one $Q$-node and one $S$-node as children. Since a segment is necessary to cover the vertices of any graph, this observation gives Corollary~\ref{cor:outer} from Theorem~\ref{th:SP}.

\begin{corollary}
\label{cor:outer} Let $G$ be a series-parallel graph with the maximum degree 3 with the SPQ-tree $\mathcal{T}$ and let $\#(P^*)$ be the number of $P$-nodes in $\mathcal{T}$ with no $Q$-node as a child. Then  $\#(P^*)+1$ segments are sometimes necessary and always sufficient to cover all vertices of $G$ in any orthogonal drawing of $G$.
\end{corollary}

\subsection{Algorithm for Rooted Trees}

Our algorithm for rooted trees in Section~\ref{sec:general} gives an algorithm for obtaining an orthogonal drawing for trees with the optimal segment-cover number, since the end point of each segment in the drawing is either a degree-1 or a degree-3 vertex; see Fig.~\ref{fig:tree}(a). However this drawing is not necessarily upward (i.e.,where each root-to-leaf path is $y$-monotone). Note that for an upward orthogonal drawing of a rooted tree, the maximum degree of a vertex is at most 3. We can obtain an upward drawing for rooted trees with maximum degree-3 by using our algorithm from Section~\ref{sec:SP}. Let $T$ be a rooted tree. Take two copies of $T$, call them $T_1$, $T_2$. Let $r_1$, $r_2$ be their respective roots. Clearly there is a trivial isomorphism between $T_1$ and $T_2$. Now identify each leaf of $T_1$ with the corresponding leaf in $T_2$. The resulting graph is then a series-parallel graph with the maximum degree 3. An upward orthogonal drawing of this graph using the algorithm from Section~\ref{sec:SP} also gives an upward drawing of $T$ (as a drawing of $T_1$ and $T_2$) with optimal segment-cover number; see Fig.~\ref{fig:tree}(b).


\section{Conclusion and Future Work}
We have provided a characterization of ``Lombardi flow'' in orthogonal drawings
in terms of minimizing the total number of segments in such a drawing or
in the number of segments covering all the vertices in such a drawing.
Given that this latter version of the problem is NP-hard, as we have shown,
an interesting direction for future work includes possible approximation
algorithms for general degree-4 planar graphs, as well as the study of other
graph classes that admit polynomial-time solutions.

\ifFull
\subsection*{Acknowledgements}
\else
\subsection{Acknowledgements}
\fi
This article reports on work supported by the Defense Advanced
Research Projects Agency under agreement no.~AFRL FA8750-15-2-0092.
The views expressed are those of the authors and do not reflect the
official policy or position of the Department of Defense or the 
U.S.~Government.
This work was also supported in part by the U.S.~National Science Foundation
under grants 1228639 and 1526631.
We would like to thank Timothy Johnson and Michael Bekos for several helpful discussions
related to the topics of this paper.


\appendix
\newpage

\section{Figures and Proof Details}
\label{sec:app-proofs}

\begin{figure}[htb]
\vspace{-0.7cm}
\centering
\includegraphics[width=0.7\textwidth]{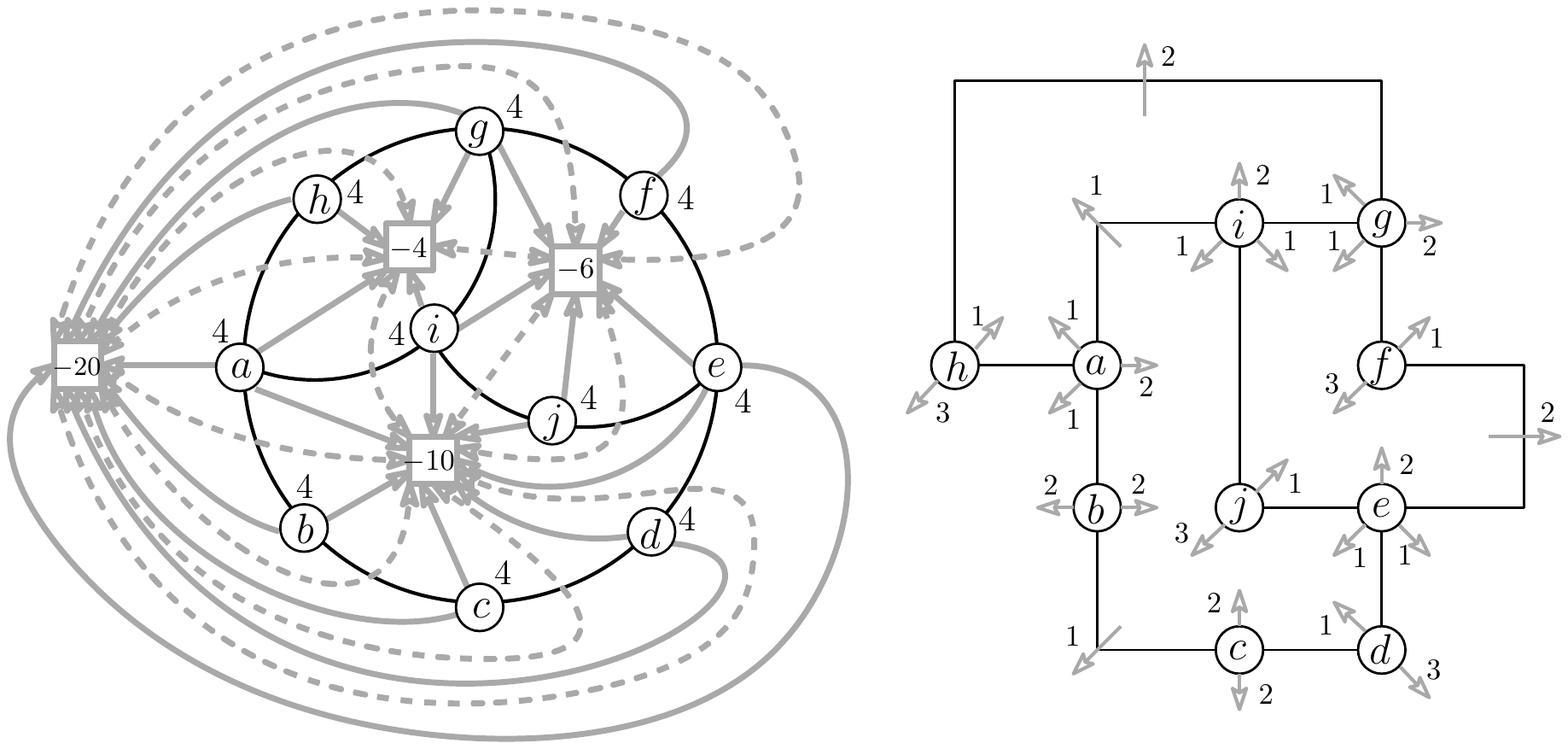}\\
(a)\hspace{0.4\textwidth}(b)
\caption{(a) Construction of the flow network (with grey edges) by
Tamassia~\cite{Tam87} for a planar graph, $G$, 
with maximum degree 4 (black edges), (b) an orthogonal drawing of $G$ and the corresponding flows in the network.}
\label{fig:tamassia}
\end{figure}

\begin{figure}[htb]
\vspace{-1.2cm}
\centering
\includegraphics[width=0.7\textwidth]{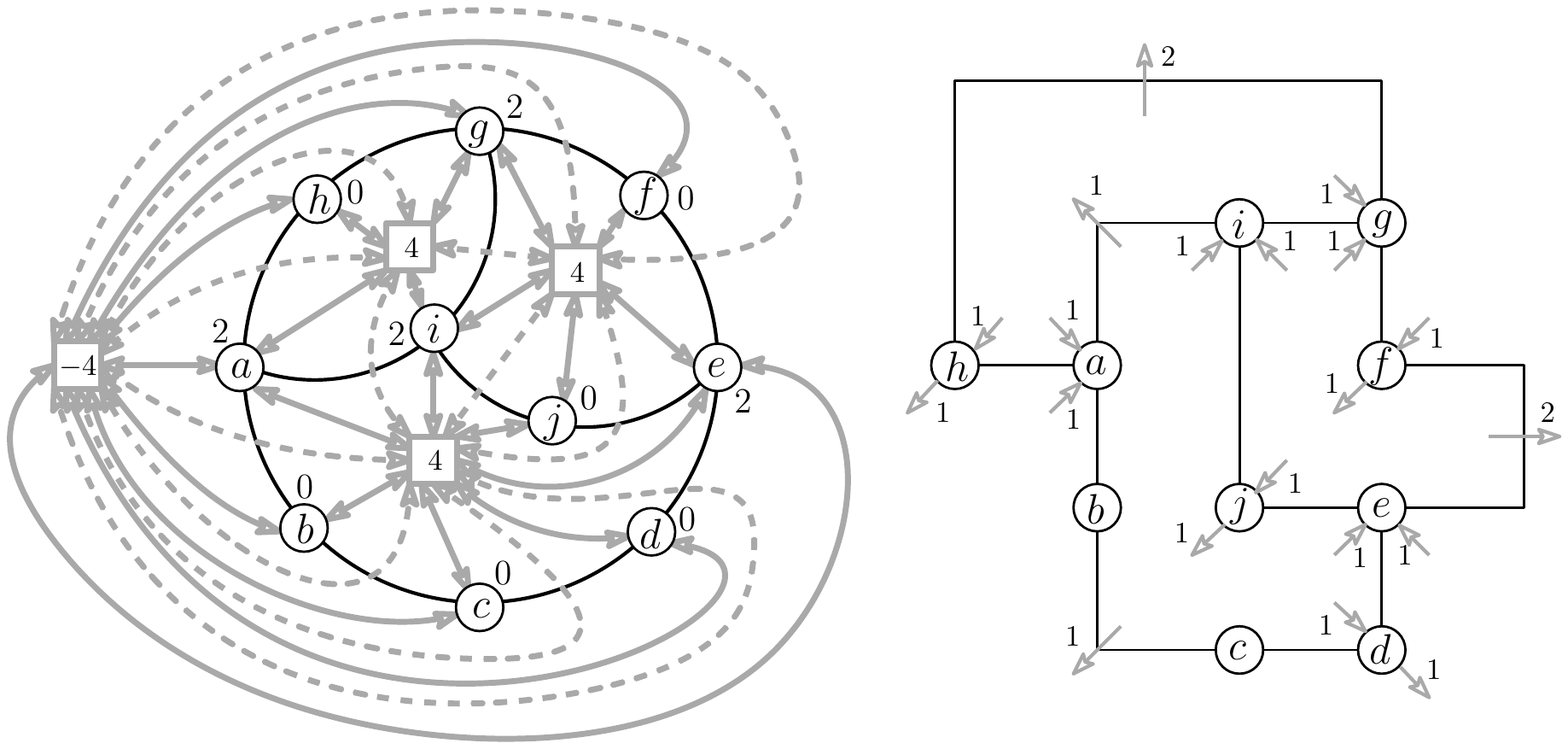}\\
(a)\hspace{0.4\textwidth}(b)
\caption{(a) Construction of the modified flow network (with grey edges) for the graph $G$ (black edges) from Fig.~\ref{fig:tamassia}(a), (b) the drawing of $G$ from Fig.~\ref{fig:tamassia}(b) and the corresponding flows in the modified network.}
\label{fig:modification}
\end{figure}

\begin{figure}[htb]
\vspace{-1.2cm}
\centering
\includegraphics[width=0.45\textwidth]{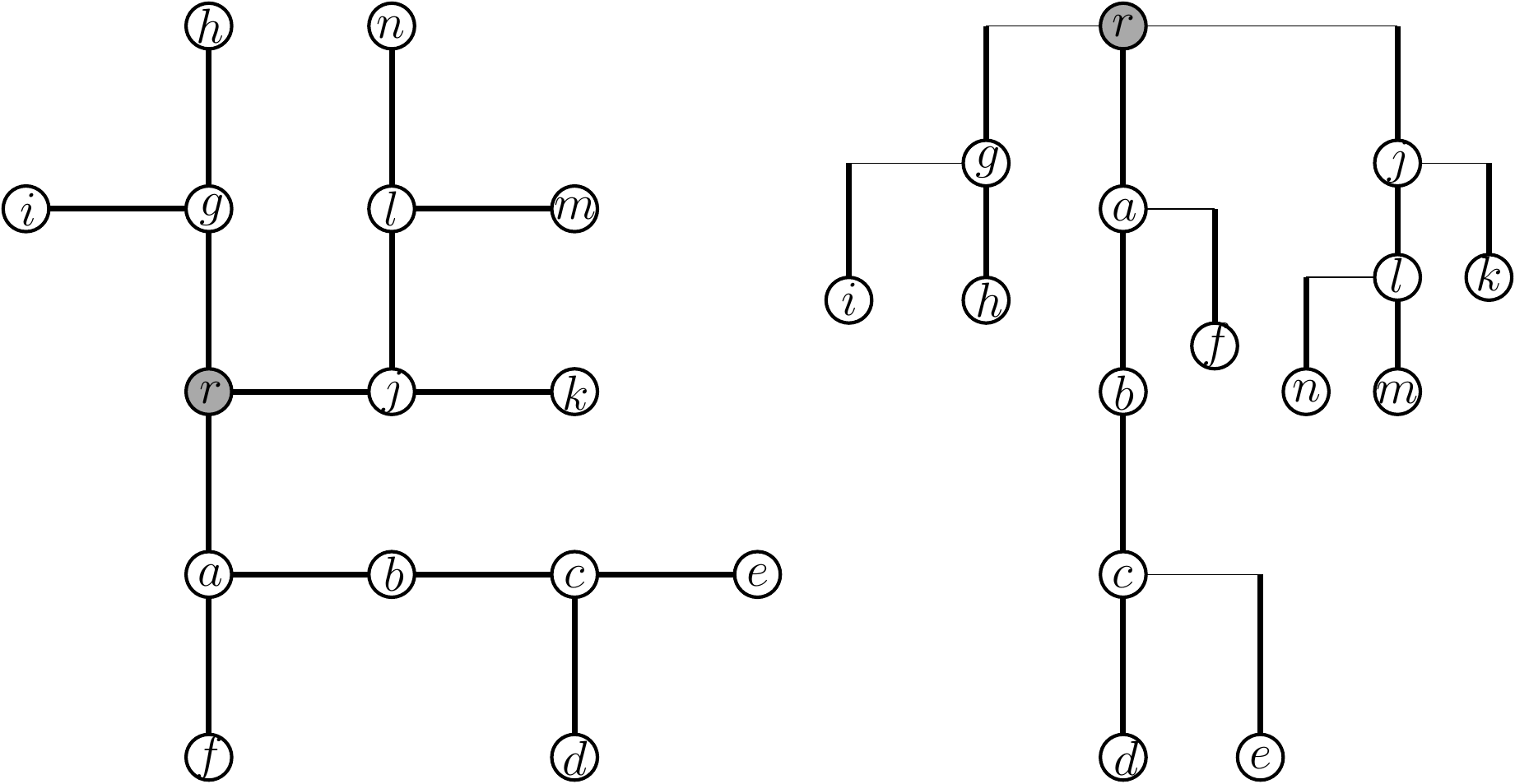}\\
(a)\hspace{0.2\textwidth}(b)
\caption{Orthogonal drawings of trees: (a) MSO-drawing and MSCO-drawing, (b) upward drawing with the minimum segment-cover number.\vspace{-0.5cm}}
\label{fig:tree}
\end{figure}

\begin{figure}[htb]
 \centering
  \subfloat[root is an $S$-node]{\includegraphics[height=0.25\textheight]{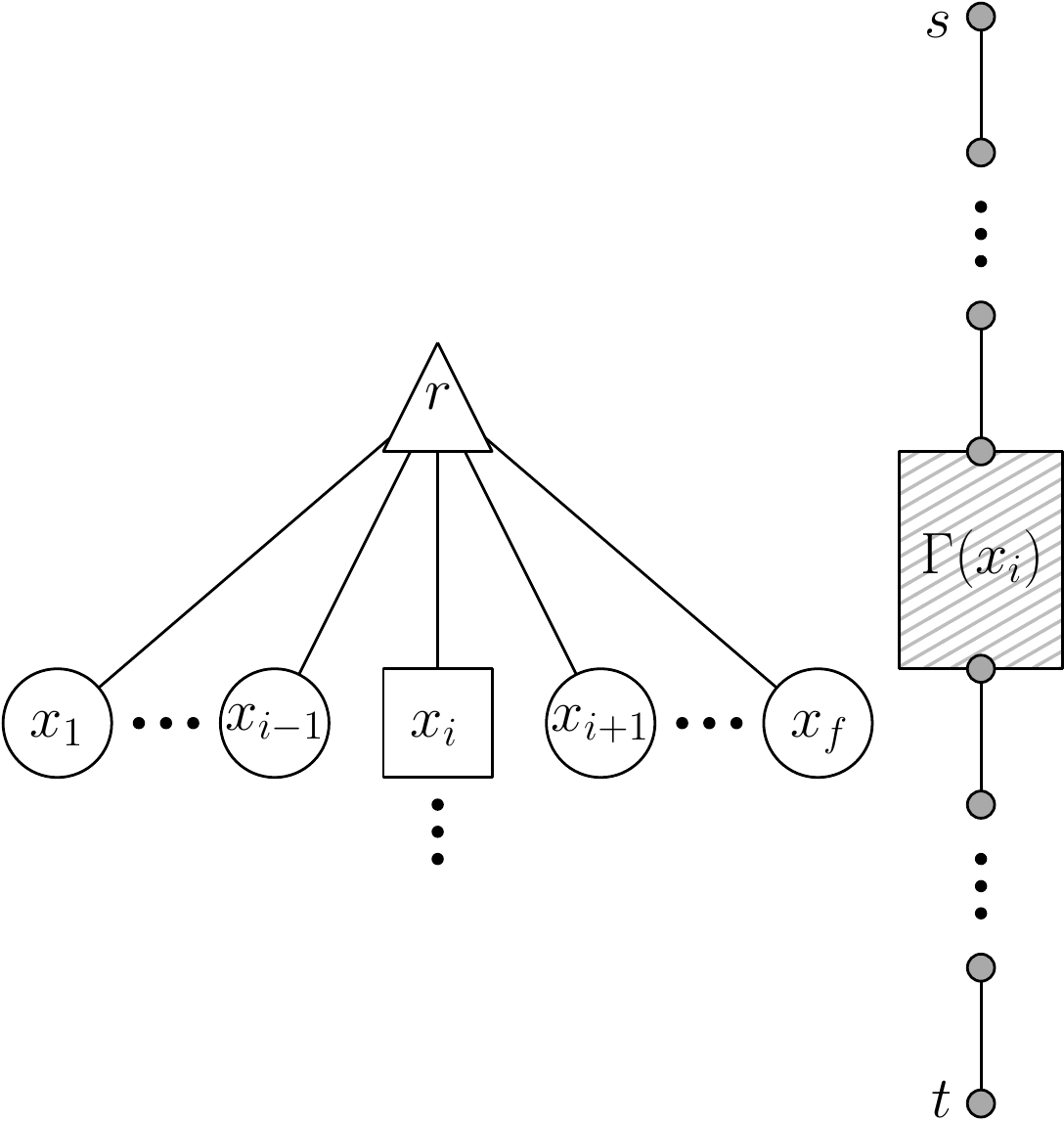}} \hfill
  \subfloat[root is a $P$-node]{\includegraphics[height=0.25\textheight]{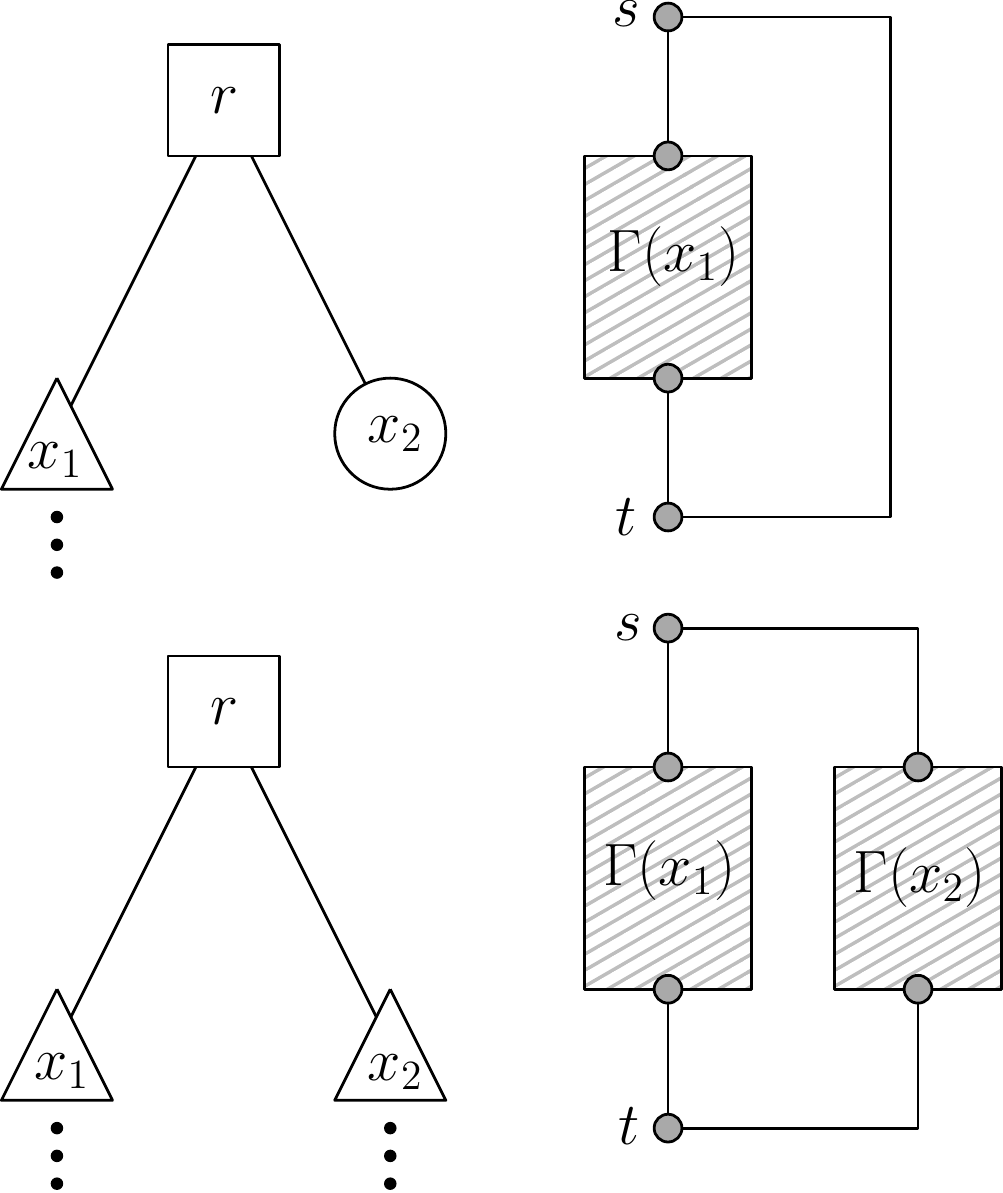}} \hfill
  \caption{Upward orthogonal drawing algorithm for a series-parallel graph. $S$-, $P$-, and $Q$-nodes are denoted by triangles, squares and circles, respectively.}
  \label{fig:SP}
\end{figure}

\newpage

\begin{figure}[htb]
\centering
\includegraphics[page=1,width=0.55\textwidth]{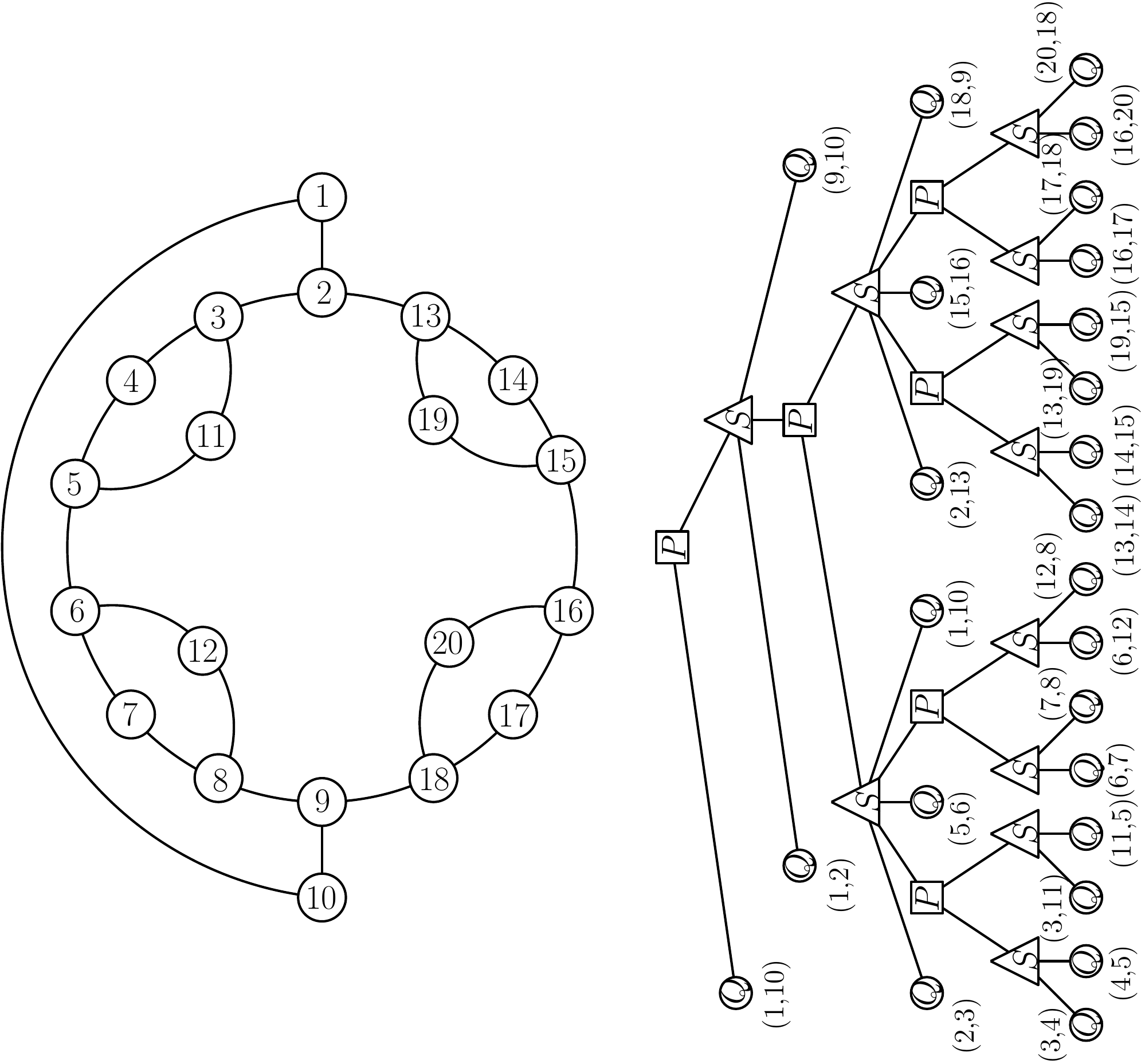}\hfill
\includegraphics[page=2,width=0.4\textwidth]{SP.pdf}\\
\hspace{0.1\textwidth}(a)\hspace{0.25\textwidth}(b)\hspace{0.2\textwidth}(c)\hspace{0.2\textwidth}(d)
\caption{(a) A series parallel graph $G$, (b) the SPQ-tree of $G$, (c) an upward drawing of $G$ with the minimum segment-cover number, (d) an MSCO-drawing of $G$.}
\label{fig:SP-illus}
\end{figure}

\newpage
\subsection{Proof of Theorem~\ref{th:augmentation}}
Let $f$ be a non-rectangular face in $\Gamma$. A vertex $v$ on $f$ is called \textit{left-open} (resp. \textit{right-open}, \textit{top-open}, \textit{bottom-open}) in $f$ if the left (resp. right, top, bottom) port of $v$ is inside $f$ and is not incident to any edge in $\Gamma$. Note that a pair of vertices $v_1, v_2$ on $f$ can align together only if one of them is left-open and the other is right open, or one of them is top open and the other bottom-open. We construct a bipartite graph $B(f)$ for $f$ as follows. The vertices of $B(f)$ are all the vertices on $f$, which are left-open, or right-open, or top-open or bottom-open. The edges of $B(f)$ are between each pair of vertices $(v_1, v_2)$ such that either (i) $v_1$ is left-open, $v_2$ is right-open (call them \textit{horizontal pair}), or (ii) $v_1$ is top-open, $v_2$ is bottom-open (call them \textit{vertical pair}). Each edge of $B(f)$ represents a potential pair of vertices which can be aligned together. However, we need to select a set of such pairs such that (i) each vertex is paired with at most one other, and (ii) the chosen pairs together can augment $f$ to a set of rectangular faces. For the two conditions to hold, the edges corresponding to all pairs in $f$ must form a maximum matching in $B(f)$ and it should be possible to place the edges for all horizontal pairs (resp. all vertical pairs) inside $f$ without crossing. Two horizontal (resp. vertical) pairs are \textit{in conflict} if both of them cannot be placed together inside $f$ without crossing. A maximum matching in $B(f)$ is called \textit{planar} if no two horizontal pairs and no two vertical pairs corresponding to these matching edges are in conflict. We have the following lemma.

\begin{lemma}
\label{lem:matching} Let $M$ be a maximum matching of $B(f)$. Then a planar maximum matching of $B(f)$ can be constructed from $M$ in polynomial time.
\end{lemma}
\begin{proof} If $M$ contains no pair of edges in conflict, then we are done. Thus consider a pair of edges  $(a,c), (b,d)$ from $M$ are in conflict. Assume without loss of generality that $a$, $b$, $c$, $d$ appear in this clockwise order around $f$. Then $M'=(M-\{(a,c),(b,d)\})\cup\{(a,d),(b,c)\}$ is another maximum matching of $B(f)$ with one fewer pairs in conflict. A repetition of this procedure eliminates all pairs of $M$ in conflict to obtain a planar maximum matching of $B(f)$. \qed
\end{proof}

Since finding a maximum matching in a bipartite graph can be done in polynomial time (e.g., using the Hopcroft-Karp algorithm~\cite{HK73}), by Lemma~\ref{lem:matching}, one can find a planar maximum matching of $B(f)$ in polynomial time. Adding the edges corresponding to the matching might create crossing between horizontal and vertical edges. However we can introduce a dummy vertex at the point of the crossing to obtain a valid rectangular layout of $f$. Doing this for each non-rectangular faces in $\Gamma$ gives an augmentation $\Gamma^*$ of $\Gamma$ into a rectangular layout. Since we maximized number of aligned pairs inside each face of $\Gamma$, $\Gamma^*$ also gives a solution to the Augmentation problem.
\qed

\medskip

\subsection{Proof of Theorem~\ref{th:general}}
Given a general planar graph $G$ with the maximum degree~4, we draw the {\closure} of $G$ using the algorithm in Section~\ref{sec:modification}, and attach the remaining tree parts to the connection vertices.

Consider a rooted tree $T$ with the maximum degree 4. We find a minimum segment drawing of $T$ as follows. We first find an arbitrary root-to-leaf path and draw the path in a single segment. Next we repetitively take a leaf $v_l$ of $T$ not already drawn, find a shortest path $P_l$ from $v_l$ to a vertex $v_p$ already drawn and draw $P_l$ as a single segment starting from a port (top, bottom, left or right) of $v_p$ not already occupied by a different edge. If $v_p$ is an odd-degree vertex in the already drawn graph, we place the segment for $P_l$ in such a way that it shares an incident segment at $v_p$ already drawn. The number of segments in such a drawing of $T$ is exactly $(\#(V_1)+\#(V_3))/2$, where $\#(V_1)$ and $\#(V_3)$ are the number of degree-1 and degree-3 vertices in $T$.

Now we describe our algorithm for a general planar graph $G$ with the maximum degree 4. We first draw the {\closure} of $G$ using the algorithm in Section~\ref{sec:modification}. Next we draw each connected component (say $T$) of the graph obtained by deleting the {\closure of $G$ from $G$} using the algorithm describe in the previous paragraph. Note that $T$ is a tree. We attach the drawing of $T$ at the corresponding connection vertex at a port not already used. Again if a connection vertex has odd degree in the already drawn graph we place the drawing of $T$ such that it shares a segment with the existing drawing.

The resulting drawing $\Gamma$ gives an embedding of the input graph $G$. The drawing $\Gamma$ is a minimum segments drawing for the embedding since the drawing of the closure minimized the vertex bends and edge bends, and the remaining parts of the graph has not vertex or edge bends. \qed

\ifFull\else
\subsection{Proof of Theorem~\ref{thm:np}}
We prove the \NP-hardness of the MSCO-drawing problem by a reduction from the
\textit{Hamiltonian path problem}~\cite{Garey:1990:CIG:574848}, 
which asks whether a given graph
$G$ has a path that contains each vertex of $G$ exactly once (i.e.,
a \textit{Hamiltonian path}). 
The Hamiltonian path problem remains
\NP-complete for planar graphs with maximum degree 3~\cite{doi:10.1137/0205049}.
Let $G$ be such a planar graph. We show that $G$ has a Hamiltonian path
if and only if its segment-cover number is 1.

Clearly, if $G$ has an orthogonal drawing that contains a segment
covering all the vertices, then the path of $G$ corresponding to
this segment is a Hamiltonian path. Conversely if $G$ has a Hamiltonian
path $P$, then we can find an orthogonal drawing of $G$ with a
single segment covering all the vertices as follows. Take a fixed
planar embedding $\emb(G)$ of $G$, draw $P$ as a single segment
(say a horizontal segment) and route the remaining edges as orthogonal
polylines with necessary bends keeping the embedding $\emb(G)$
fixed. This gives a valid orthogonal drawing since each vertex has
a maximum degree 3, and hence from each vertex $v$, its incident
edges on $P$ are drawn to its left and right, while the only other
edge (if any) is drawn to its top or bottom without any overlap.
\qed
\fi

\newpage

\section{Examples for MSCO-Drawings}

In this section we illustrate some examples of an orthogonal drawing obtained from the algorithm  by Tamassia~\cite{Tam87}, along with an MSO-drawing, and an MSCO-drawing of the same graph.

\begin{figure}[htb]
\centering
\includegraphics[page=1,width=0.8\textwidth]{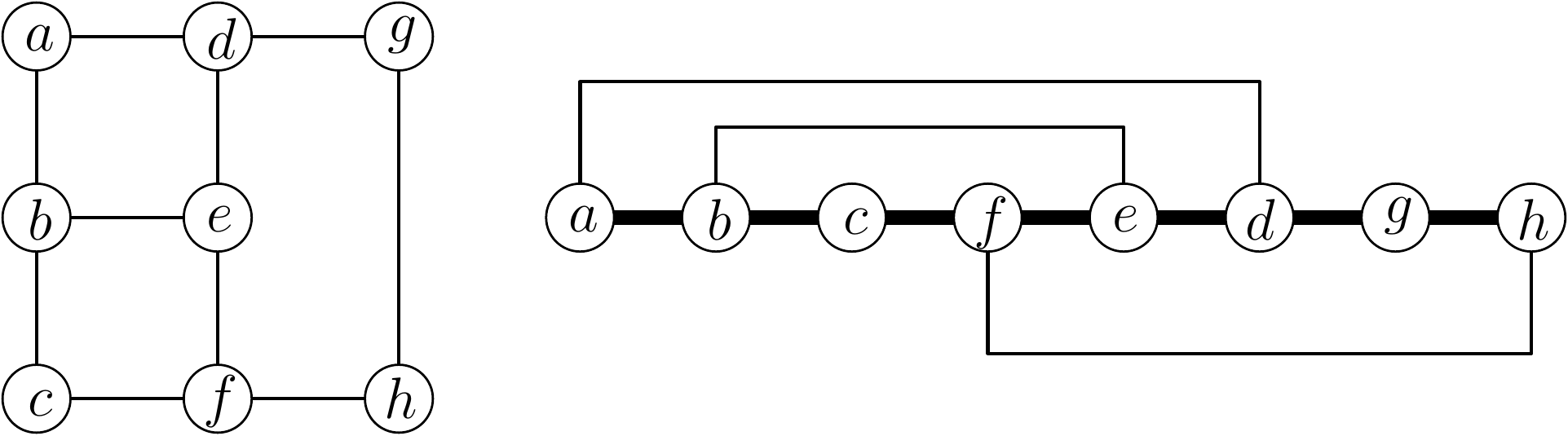}\\
(a)\hspace{0.4\textwidth}(b)\hspace{0.15\textwidth}
\caption{(a) A drawing of a planar graph $G$ using the algorithm by Tamassia~\cite{Tam87}, which is also a rectangular drawing, (b) an MSCO-drawing of $G$ with only a single segment covering all the vertices.}
\label{fig:example01}
\end{figure}

\begin{figure}[htb]
\centering
\includegraphics[page=2,width=0.8\textwidth]{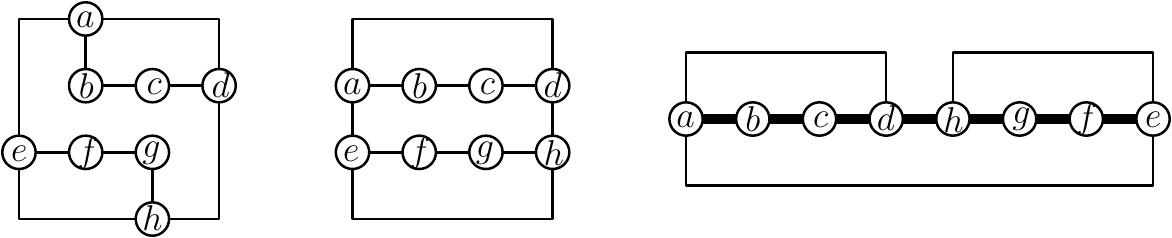}\\
(a)\hspace{0.18\textwidth}(b)\hspace{0.27\textwidth}(c)\hspace{0.15\textwidth}
\caption{(a) A drawing of a planar graph $G$ using the algorithm by Tamassia~\cite{Tam87}, (b) an MSO-drawing of $G$, (c) an MSCO-drawing of $G$ with only a single segment covering all the vertices.}
\label{fig:example02}
\end{figure}

\begin{figure}[htb]
\centering
\includegraphics[page=3,width=0.8\textwidth]{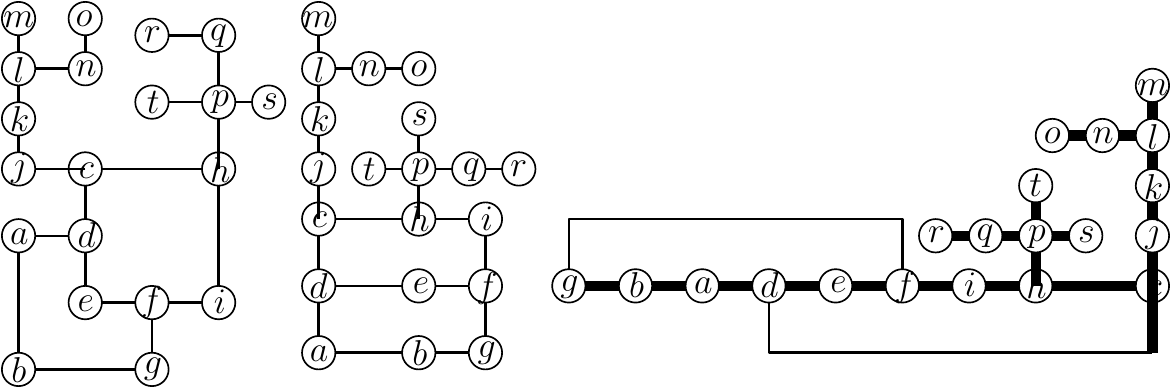}\\
(a)\hspace{0.18\textwidth}(b)\hspace{0.27\textwidth}(c)\hspace{0.1\textwidth}
\caption{(a) A drawing of a planar graph $G$ using the algorithm by Tamassia~\cite{Tam87}, (b) an MSO-drawing of $G$, (c) an MSCO-drawing of $G$.} 
\label{fig:example03}
\end{figure}

\begin{figure}[htb]
\centering
\includegraphics[page=4,width=0.8\textwidth]{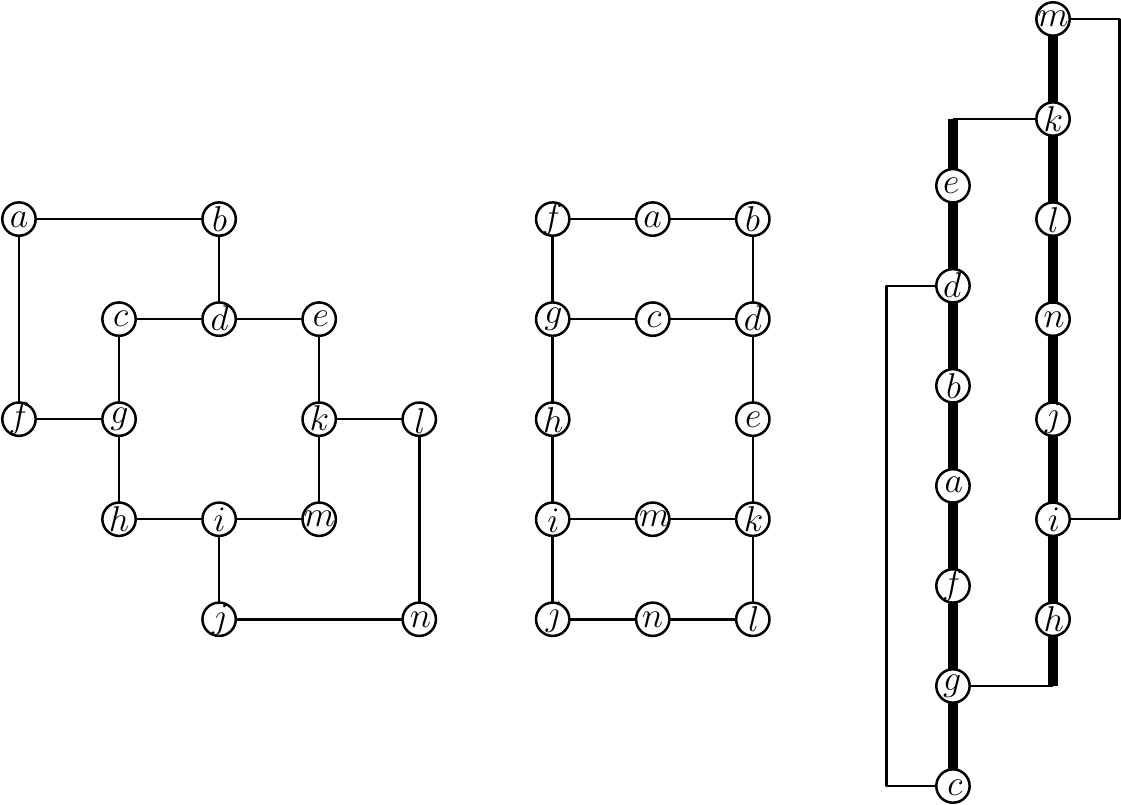}\\
\hspace{0.05\textwidth}(a)\hspace{0.25\textwidth}(b)\hspace{0.25\textwidth}(c)
\caption{(a) A drawing of a planar graph $G$ using the algorithm by Tamassia~\cite{Tam87}, (b) an MSO-drawing of $G$, (c) an MSCO-drawing of $G$ with only a single segment covering all the vertices.}
\label{fig:example04}
\end{figure}

\begin{figure}[htb]
\centering
\includegraphics[page=5,width=0.8\textwidth]{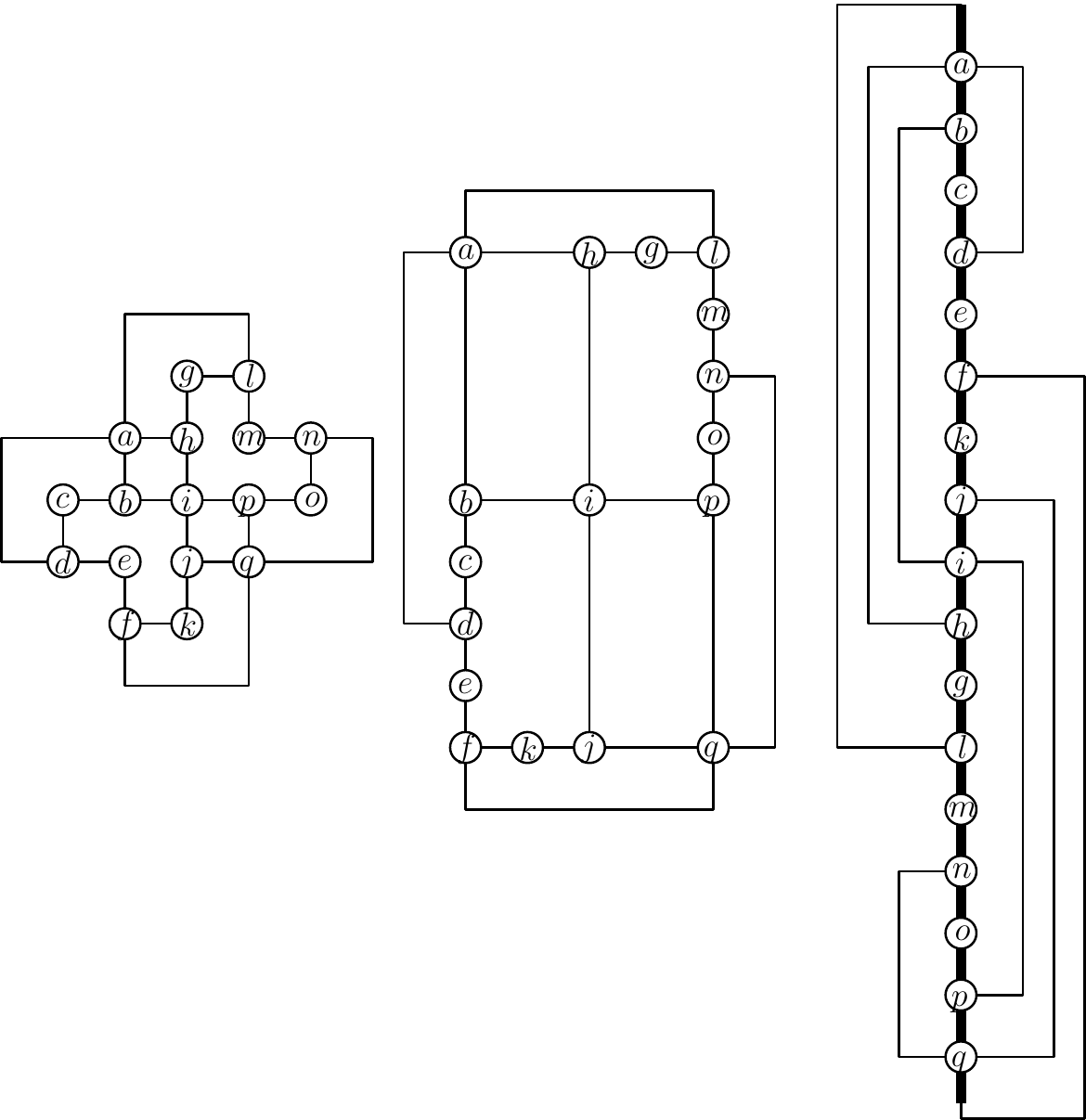}\\
\hspace{0.05\textwidth}(a)\hspace{0.25\textwidth}(b)\hspace{0.25\textwidth}(c)
\caption{(a) A drawing of a planar graph $G$ using the algorithm by Tamassia~\cite{Tam87}, (b) an MSO-drawing of $G$, (c) an MSCO-drawing of $G$ with only a single segment covering all the vertices.}
\label{fig:example05}
\end{figure}

\end{document}